\documentclass[12pt]{article}

\usepackage{amsmath}
\usepackage{amsthm}

\newcommand{\half}{\frac{1}{2}}
\newcommand{\q}{q}
\newcommand{\db}{\Omega}
\newcommand{\K}{K}
\newcommand{\V}{\tilde V}
\newcommand{\e}{\beta}
\newcommand{\G}{{ \cal G}}

\newtheorem{Theorem}{Theorem}
\newtheorem{Corollary}{Corollary}
\newtheorem{Proposition}{Proposition}

\date{May 27, 2015}

\numberwithin{equation}{section}

\begin{document}

\begin{title}
{\bf Geometrization Conditions for Perfect Fluids, Scalar Fields, and Electromagnetic Fields}
\end{title}

\author{
D.~S.~Krongos
and
C.~G.~Torre
\\
\it Department of Physics, 
Utah State University\\ 
\it Logan, UT, USA, 84322-4415
}

\maketitle

\begin{abstract}
Rainich-type conditions giving a spacetime ``geometrization'' of matter fields in general relativity are reviewed and extended.   Three types of matter  are considered: perfect fluids,  scalar fields, and electromagnetic fields.  Necessary and sufficient conditions on a spacetime metric for it to be part of a perfect fluid solution of the Einstein equations are given.  Formulas for constructing the fluid from the metric are  obtained. All fluid results hold for any spacetime dimension.  Geometric conditions on a metric which are necessary and sufficient for it to define a solution of the Einstein-scalar field equations and formulas for constructing the scalar field from the metric are unified and extended to arbitrary dimensions, to include a cosmological constant, and to include any self-interaction potential.  Necessary and sufficient conditions on a four-dimensional spacetime metric for it to be an electrovacuum and formulas for constructing the electromagnetic field from the metric are generalized to include a cosmological constant.  Both  null and non-null electromagnetic fields are treated.   A number of  examples and applications of these results are presented.

\end{abstract}
\bigskip\bigskip

\section{Introduction}

In the general theory of relativity it is often possible to eliminate the matter fields from the Einstein-matter field equations and express the equations as local geometric conditions on the spacetime metric alone.  This possibility was discovered by Rainich \cite{Rainich}, who showed how to eliminate  the Maxwell field from the Einstein-Maxwell equations, arriving at the ``Rainich conditions", which  give necessary and sufficient  conditions on a spacetime metric for it  to be a non-null electrovacuum.   Rainich's work was made prominent by Misner and Wheeler \cite{MW}, who advanced the  ``geometrization'' program in which all matter was to be modeled as a manifestation of spacetime geometry.   Over the subsequent years a variety of additional geometrization results have been found pertaining to electromagnetic fields, scalar fields, spinor fields, fluids, and so forth.  See, {\it e.g.}, references \cite{Peres}--\cite{Torre2014}.    Results such as these provide, at least in principle, a new way to analyze field equations and their solutions from a purely geometric point of view, just involving the metric.  

The geometrization conditions which have been obtained over the years, while conceptually  elegant, are generally  more complicated than the original Einstein-matter field equations.  For example, the Einstein-Maxwell equations are a system of variational second-order PDEs, while the Rainich conditions involve  a system of non-variational fourth-order PDEs.  For this reason one  can understand why geometrization results have seen relatively little practical use in relativity and field theory. The current abilities of symbolic computational systems  have, however, made the use of geometrization conditions viable for various applications.  Indeed, the bulk of the non-null electrovacuum solutions presented in the treatise of reference \cite{Stephani} were verified using a symbolic computational implementation of the classical Rainich conditions.  The geometrization examples for various null electrovacua found in \cite{Torre2014} were obtained using the {\sl DifferentialGeometry} package in {\it Maple} \cite{AT2012}.  The purpose of this paper is to compile a set of geometrization results for the Einstein field equations which involve the most commonly used matter fields and  which are as comprehensive and as general as possible while at the same time in a form suitable for symbolic computational applications.  

This last point requires some elaboration as it significantly constrains the type of geometric conditions which we shall deem suitable for our purposes.  
A suitable geometrization condition for our purposes will define an algorithm which takes as input a given spacetime metric and which determines, solely through algebraic and differentiation operations on the metric, whether the metric defines a solution of the Einstein equations with a given matter content.  When the metric does define a solution, the matter fields shall be constructed  directly from the metric via  algebraic operations, differentiation, and integration.
 
  We summarize our treatment for each of  three types of matter fields and compare to existing results as follows.

\medskip\noindent
{\sl Perfect Fluids}

We  give necessary and sufficient conditions on a spacetime metric for it to be part of a perfect fluid solution of the Einstein equations.  Formulas for constructing the fluid from the metric are  obtained.  These results apply to spacetimes of any dimension greater than two and allow for a cosmological constant.  No energy conditions or equations of  state are imposed.    Existing geometrization conditions for fluids can be found in \cite{Senovilla}, which generalize those found in four spacetime dimension in \cite{CollFerrando}.   The conditions given in \cite{Senovilla} and \cite{CollFerrando}, while elegant, involve the existence of certain unspecified functions and so do not  satisfy the computational criteria listed above.  The results of \cite{CollFerrando} also include conditions which enforce equations of state, while the results of \cite{Senovilla} enforce the dominant energy condition. Our conditions enforce neither of these since we are interested in geometrization conditions which characterize any type of fluid solutions. The geometrization conditions we obtain are built algebraically from the Einstein tensor and so involve up to two derivatives of the metric.

\medskip\noindent
{\sl Scalar Fields}

 Geometric conditions on a metric which are necessary and sufficient for it to define a solution of the Einstein-scalar field equations and formulas for constructing the scalar field from the metric have been obtained by Kucha\v r for free fields in four spacetime dimensions without a cosmological constant \cite{Kuchar}. These results apply to massless and massive fields.  The results of \cite{Kuchar} subsume related results in  \cite{Peres} and \cite{Penney1965}.  More recently,  conditions for a symmetric tensor to be algebraically that of a free massless scalar field in any dimension, without cosmological constant, have been given in \cite{Senovilla}. These conditions are necessary algebraically but are not sufficient for geometrization since additional differential conditions are required.  Here we give necessary and sufficient conditions for a metric to define a solution of the Einstein-scalar field equations which generalize all these results.  In particular, the results we obtain here  hold in arbitrary dimensions, they allow for a cosmological constant,  they allow for a mass, and they allow for a freely specifiable self-interaction potential.  Null and non-null fields are treated. The geometrization conditions we have found for a scalar field necessarily involve both algebraic and differential conditions on the Einstein tensor; they involve up to three derivatives of the spacetime metric. 
 
 \medskip\noindent
{\sl Electromagnetic Fields} 

Necessary and sufficient conditions on a four-dimensional spacetime metric for it to be a non-null electrovacuum were given by Rainich \cite{Rainich} and Misner, Wheeler \cite{MW}.  The null case has been investigated by Misner, Wheeler \cite{MW}, Peres \cite{Peres1961}, Geroch \cite{Geroch},  Bartrum \cite{Bartrum},  and Ludwig \cite{Ludwig}.  Building upon Ludwig's results, a set of geometrization conditions for null electrovacua  has been given by one of us in reference \cite{Torre2014} via the Newman-Penrose formalism.  Here we generalize both the non-null results of Rainich, Misner and Wheeler and the null results of reference \cite{Torre2014}  to include a cosmological constant.  While this represents a very modest generalization of the classical Rainich conditions in the non-null case, it represents a less obvious generalization in the null case. In any case, we hope there is some value in assembling all these results in one place and in a unified form which is amenable to computational algorithms. All these results hold in four dimensions. The geometrization conditions for non-null electromagnetic fields involve up to four derivatives of the metric, while the geometrization of null electromagnetic fields involves as many as five derivatives of the metric. 

\medskip
Besides proving various geometrization theorems, we provide a number of modest illustrations of the theorems which hopefully serve to clarify their structure and usage.  All these illustrations were accomplished using the {\sl DifferentialGeometry} package in {\it Maple}, amply demonstrating the amenability of our results to symbolic computation. Software implementation of our geometrization results, computational details of some of the applications presented here, along with additional applications, can be found at {\tt http://digitalcommons.usu.edu/dg/}.

\section{Perfect Fluids}

Let $(M, g)$ be an $n$-dimensional spacetime, $n > 2$, with signature $(-++ \dots +)$.  Let $\mu\colon M\to {\bf R}$ and $p\colon M\to {\bf R}$ be functions on $M$.  Let $u$ be a  unit timelike vector field on $M$, that is, $g_{ab}u^a u^b = - 1$. The Einstein equations for a perfect fluid are 
\begin{equation}
R_{ab} - \frac{1}{2} R g_{ab} + \Lambda g_{ab} = q\Big\{(\mu + p) u_a u_b + p g_{ab}\Big\}.
\end{equation}
Here $R_{ab}$ is the Ricci tensor of $g_{ab}$, $R = g^{ab}R_{ab}$ is the Ricci scalar,  $\Lambda$ is the cosmological constant, and $q = {8\pi \G}/{c^4}$ with $\G$ being Newton's constant.  We note that the cosmological and Newton constants can be absorbed into the definition of the fluid.  With
\begin{equation}
\tilde\mu = q\mu + \Lambda,\quad \tilde p = qp - \Lambda,
\end{equation}
the Einstein equations take the form
\begin{equation}
R_{ab} - \frac{1}{2} R g_{ab}  =  (\tilde\mu +\tilde p) u_a u_b + \tilde p g_{ab}.
\label{EEPF}
\end{equation}

If there exist functions $\tilde \mu$ and $\tilde p$ and a timelike unit vector field $u$ on $M$ such that (\ref{EEPF}) holds, we say that $(M, g)$ is a perfect fluid spacetime.  Note that if $\tilde \mu + \tilde p = 0$ in some open set ${\cal U}\subset M$ then the spacetime is actually an Einstein space on $\cal U$. In what follows, when we speak of perfect fluid spacetimes we assume that $\tilde \mu + \tilde p \neq 0$ at each point of the spacetime. Note also that we have not imposed any energy conditions, equations of state, or thermodynamic properties.  These additional considerations are examined, for example, in reference \cite{CollFerrando}.

In the following we will use the  trace-free Ricci (or trace-free Einstein) tensor $S_{ab}$:
\begin{equation}
S_{ab} = R_{ab} - \frac{1}{n} R g_{ab} = G_{ab} - \frac{1}{n} G g_{ab}.
\end{equation}
We will also need the following elementary result.
\begin{Proposition}
Let $Q_{ab}$ be a covariant, symmetric, rank-2 tensor on an $n$-dimensional vector space $V$. $Q_{ab}$ satisfies 
\begin{equation}
	 Q_{a[b} Q_{c]d}  = 0
\label{Qeq}
\end{equation}
if and only if there exists a covector $v_a\in V^*$ such that 
\begin{equation} 
	 Q_{ab} =\pm v_a v_b.
\label{vveq}
\end{equation}
\label{prodvec}
\end{Proposition}

\begin{proof}
Eq. (\ref{vveq}) clearly implies (\ref{Qeq}).  We now show that (\ref{Qeq}) implies (\ref{vveq}). From Sylvester's law of inertia there exists a basis for $V^*$, denoted by $\e_i, i=1,2,\dots,n$,  in which  $Q_{ab}$ is diagonal with components given by $\pm1$, $0$. In this basis, using index-free notation, $Q$ takes the form:
\begin{equation}
	Q = \sum_{i=1}^n a_i \e_i \otimes \e_i,
\end{equation}
where $a_i \in \{-1, 0, 1\}$. In this basis, eq. (\ref{Qeq}) takes the form
\begin{equation}
	 \sum_{i, j = 1}^n a_i a_j (\e_i \otimes \e_i \otimes \e_j \otimes \e_j - \e_i \otimes \e_j \otimes \e_i \otimes \e_j)\\
	= 0.
\end{equation}
Consequently, $a_i a_j = 0$ for all $i \ne j$, from which it follows that all but one of the $\{a_i\}$ are zero (assuming $Q\neq 0$).  If the basis is labeled so that $a_1$ is the non-zero component then $v = \sqrt{|a_1|}\e_1$, and (\ref{vveq}) follows. 
\end{proof}

The following theorem gives a simple set of Rainich-type conditions  for a perfect fluid spacetime. 
\begin{Theorem}
 Let $(M, g)$ be an $n$-dimensional  spacetime, $n>2$. Define
 \begin{equation}
	\alpha = -\left[\frac{n^2}{(n-1)(n-2)} S_a^bS_b^cS_c^a\right]^{1/3}.
\label{alphaeq}
\end{equation}  
The metric $g$ defines a perfect fluid spacetime if and only if 
\begin{align}
&(1)  \quad \alpha \neq 0,\\
&(2)  \quad \K_{a[b} \K_{c]d} = 0,\\
&(3) \quad \K_{ab}v^av^b > 0,\ {\rm for\ some\ } v^a,
\label{pfcond}
\end{align}
where  
\begin{equation}
\K_{ab} = \frac{1}{\alpha}S_{ab} - \frac{1}{n}g_{ab}.
\label{heq}
\end{equation}
\label{PFThm}
\end{Theorem}

\begin{proof}
We begin by showing the conditions are necessary. Suppose the Einstein equations (\ref{EEPF}) are satisfied for some $(g, \tilde \mu, \tilde p, u)$. The trace-free Einstein tensor takes the form
\begin{equation}
	S_{ab}  = (\tilde \mu + \tilde p) \left( u_a u_b +\frac{1}{n} g_{ab} \right).
\label{tfeepf}
\end{equation}
Equations (\ref{tfeepf}) and (\ref{alphaeq}) yield
\begin{equation}
	\alpha = \tilde\mu + \tilde p,
	\label{alph}
\end{equation}
which implies condition (1) since we are always assuming that $\tilde \mu+ \tilde p \neq 0$.
From (\ref{alph}), (\ref{heq}), and (\ref{EEPF}) we have that
\begin{equation}
\quad \K_{ab} = u_a u_b,
\end{equation}
which implies conditions (2) and (3).

We now check that conditions (1)--(3) are sufficient.  Condition (1) permits $\K_{ab}$ to be defined. From Proposition \ref{prodvec}, condition (2) implies there exists a covector field $u_a$ such that $\K_{ab} = \pm u_a u_b$, while condition (3) picks out the positive  sign, $\K_{ab} =  u_a u_b$.  From the definition (\ref{heq}) of $\K_{ab}$ it follows that $g_{ab}u^a u^b = -1$.   Defining $\tilde\mu$ and $\tilde p$ via
\begin{equation}
\tilde \mu + \tilde p = \alpha, \quad  \tilde p = \frac{1}{n} (G + \alpha),
\end{equation}
it follows the perfect fluid Einstein equations (\ref{EEPF}) are satisfied. 
\end{proof}

From the proof of this theorem we obtain a prescription for construction of the fluid variables from a metric satisfying the conditions (1)--(3).
\begin{Corollary}
If $(M, g)$ is an $n$-dimensional spacetime satisfying the conditions of Theorem \ref{PFThm} then it is a perfect fluid spacetime with energy density $\tilde\mu$ and pressure $\tilde p$ given by
\begin{equation}
\tilde \mu = \frac{1}{n}[(n-1) \alpha - G],\quad \tilde p = \frac{1}{n}(\alpha + G),
\end{equation}
and  fluid velocity $u^a$ determined (up to an overall sign) from the quadratic condition
\begin{equation}
u_a u_b = \K_{ab}.
\end{equation}
\label{PFCor}
\end{Corollary}

\subsection{Example: A  static, spherically symmetric perfect fluid}

Here we use  Theorem \ref{PFThm} to find fluid solutions.  
Consider the following simple ansatz for a class of static, spherically symmetric spacetimes:
\begin{equation}
g = -r^2 dt \otimes dt + f(r) dr \otimes dr + r^2(d\theta \otimes d\theta + \sin^2\theta d\phi \otimes d\phi).
\label{fluid1}
\end{equation}
Here $f$ is a function to be determined. Computation of the tensor field $K$  and imposition the quadratic condition $K_{a[b}K_{c]d} = 0$ leads to a system of non-linear ordinary differential equations for $f(r)$ which can be reduced to:
\begin{equation}
r f^\prime  + 2f- f^2 = 0.
\end{equation}  
This has the 1-parameter family of solutions:
\begin{equation}
f(r) = \frac{2}{1 + \lambda r^2},
\label{fluid1b}
\end{equation}
where $\lambda$ and $r$ are restricted by $1 + \lambda r^2 > 0$ to give the metric  Lorentz signature.  With $f(r)$ so determined, the scalar $\alpha$ and the tensor $K$ are computed to be
\begin{equation}
\alpha = \frac{1}{r^2},\quad K = r^2 dt \otimes dt,
\end{equation}
from which it immediately follows that all 3 conditions of Theorem \ref{PFThm} are satisfied. 

From Corollary \ref{PFCor} the energy density, $\tilde \mu$, pressure $\tilde p$ and 4-velocity $u$ are given by
\begin{equation}
\tilde \mu = \frac{1}{2}[\frac{1}{r^2} - 3\lambda],\quad
\tilde p = \frac{1}{2}[\frac{1}{r^2} + 3\lambda],\quad
u = \frac{1}{r}\partial_t.
\end{equation}
If desired, one can interpret this solution as admitting a cosmological constant $\Lambda = - \frac{3}{2}\lambda$ and a stiff equation of state, $\mu = p = 1/(2qr^2)$. 


\subsection{Example: A class of 5-dimensional cosmological fluid solutions}

In this example we use  Theorem \ref{PFThm} to construct a class of cosmological perfect fluid solutions on a 5-dimensional spacetime $(M, g)$ where $M = {\bf R}\times \Sigma^4$ with $\Sigma^4 = {\bf R}^3 \times {\bf R}$ being homogeneous and anisotropic. We start from a 4-parameter family of metrics of the form
\begin{equation}
g = - dt \otimes dt + r_0^2 t^{2b} (dx \otimes dx + dy\otimes dy + dz\otimes dz) + R_0^2 t^{2\beta} dw \otimes dw,
\end{equation}
where $r_0$, $R_0$, $b$, and $\beta$ are parameters to be determined.  This metric defines a family of spatially flat 3+1 dimensional FRW-type universes with  $(x, y, z)$ coordinates, each with an extra dimension $w$ described with its own scale factor.  Using  Theorem \ref{PFThm} we select metrics from this set which solve the perfect fluid Einstein equations. 

Condition (2) of Theorem \ref{PFThm} leads to a system of algebraic equations for $b$ and $\beta$ from which we have found 3  solutions:
\begin{equation}
{\rm (i)}\ b = \frac{\beta(\beta-1)}{\beta+2}, \quad {\rm (ii)}\ b = -\frac{1}{3}(\beta-1),\quad {\rm (iii)}\  b = \beta.
\end{equation} 
Case (i) can be eliminated from consideration since
\begin{equation}
b = \frac{\beta(\beta-1)}{\beta+2}\quad \Longrightarrow\quad K = -R_0^2 t^{2\beta} dw \otimes dw,
\end{equation}
which violates condition (3) of Theorem \ref{PFThm}.  Cases (ii) and (iii) satisfy all the conditions of Theorem \ref{PFThm}. Using Corollary \ref{PFCor} we obtain solutions of the Einstein equations as follows:
\begin{align}
&{\rm (ii)}\ b = -\frac{1}{3}(\beta-1)\quad \Longrightarrow\quad K = dt\otimes dt,\quad u = \partial_t,\quad \tilde\mu=\tilde p=-\frac{1}{3}\frac{2\beta^2-\beta-1}{t^2}\\
&{\rm (iii)}\ b = \beta\quad \Longrightarrow\quad K = dt\otimes dt,\quad u = \partial_t\quad \tilde \mu = \frac{6\beta^2}{t^2}, 
\tilde  p =-\frac{3\beta(2\beta-1)}{t^2}.
\end{align} 
Case (ii) is anisotropic (except when $b = \beta = \frac{1}{4}$) and allows for any combination of expansion and contraction for the $(x, y, z)$ and $w$ dimensions. Case (iii) is isotropic in all four spatial dimensions. Both cases are singular as $t\to 0$. 

\section{Scalar Fields}

Kucha\v r has given Rainich-type geometrization conditions for a minimally-coupled, free scalar field in four spacetime dimensions without a cosmological constant \cite{Kuchar}.  Here we generalize his treatment to a real scalar field with any self-interaction, in any dimension, and including the possibility of a cosmological constant.

The Einstein-scalar field equations for a spacetime $(M, g)$ with a minimally coupled real scalar field $\psi$, with self interaction potential $V(\psi)$, and with cosmological constant $\Lambda$ are given by
\begin{equation}
	R_{ab} - \frac{1}{2}R g_{ab} + \Lambda g_{ab}= q\left\{\psi_{;a} \psi_{;b} - \frac{1}{2}  g^{lm} \psi_{;l} \psi_{;m}g_{ab} - V(\psi)  g_{ab}\right\},
\label{esf}
\end{equation}
\begin{equation}
	g^{lm}\psi_{;lm} - V^\prime(\psi) = 0,
\label{kg}
\end{equation}
where  $q = {8\pi \G}/{c^4}$ and we use a semicolon to denote the usual torsion-free, metric-compatible covariant derivative.   

We distinguish two classes of solutions to the Einstein-scalar field equations.  If a solution has $\psi_{;a}\psi_;{}^a \neq 0 $ everywhere we say that the solution is {\it non-null}.  If the solution has $\psi_{;a}\psi_;{}^a = 0 $ everywhere we say that the solution is {\it null}.

The Rainich-type conditions we shall obtain require the following extension of Proposition \ref{prodvec} ({\it c.f.} \cite{Kuchar}).
\begin{Proposition}
Let $Q_{ab}$ be a symmetric tensor field on a manifold $M$.  Then (locally on $M$) there exists a function $\phi$ such that
\begin{equation}
Q_{ab} = \pm \phi_{;a} \phi_{;b}
\label{Q2eq}
\end{equation}
if and only if $Q_{ab}$ satisfies 
\begin{align}
	&(1)\ Q_{a[b}Q_{c]d} =0,\\
	&(2)\ Q_{ab}Q_{c[d;e]} + Q_{ac}Q_{b[d;e]} + Q_{bc;[d}  Q_{e]a}= 0.
\end{align}
\begin{proof}
Using Proposition \ref{prodvec}, condition (1) is necessary and sufficient for the existence of a 1-form $v_a$ such that $Q_{ab} = \pm v_a v_b$. We then have 
\begin{equation}
	Q_{bc;[d}  Q_{e]a} = v_a v_b v_{c;[d}v_{e]} + v_a v_c  v_{b;[d} v_{e]},
\end{equation}
and  
\begin{equation}
	Q_{c[d;e]} = \pm( v_{c;[e} v_{d]}  + v_c v_{[d;e]}),
\end{equation}
and hence
\begin{equation}
	Q_{ab}Q_{c[d;e]} + Q_{ac}Q_{b[d;e]} + Q_{bc;[d}Q_{e]a} = 2v_a v_b v_c v_{[d;e]}  .
\end{equation}
Consequently, if (\ref{Q2eq}) holds then condition (2) holds and, conversely, condition (2) implies  the 1-form $v_a$ is closed, hence locally exact.
\label{curlgrad}
\end{proof}
\end{Proposition}

We shall use the following notation:
\begin{equation}
G_{ab} = R_{ab} - \frac{1}{2}Rg_{ab},\quad G = G_a^a,\quad {}_2G_{ab} = G_a^cG_{cb},\quad {}_2G = G_{ab}G^{ab},
\end{equation}
\begin{equation}
{}_3G_{ab} = G_a^cG_c^dG_{db},\quad {}_3G =  G_a^cG_c^dG_{d}^a.
\end{equation}


The geometrization theorems we shall obtain will depend upon whether the self-interaction potential $V(\psi)$ is present.  We begin with a free, massless, non-null scalar field.  
\begin{Theorem}
Let $(M, g)$ be an $n$-dimensional spacetime, $n>2$. The following are  necessary and sufficient conditions on $g$ such that there exists a scalar field $\psi$  with $(g, \psi)$ defining a local, non-null solution to the Einstein-scalar equations (\ref{esf}), (\ref{kg}) with $V(\psi) = 0$:
\begin{equation}
 {}_2G - \frac{1}{n}G^2 \neq 0,
 \label{sf5}
 \end{equation}
 \begin{equation}
 A_{;i}=0,
 \label{sf0}
 \end{equation}
\begin{equation}
	H_{a[b}H_{c]d} = 0,
\label{sf1}
\end{equation}
\begin{equation}
	 H_{ab}H_{c[d;e]} + H_{ac}H_{b[d;e]} + H_{bc;[d}H_{e]a} = 0, 
\label{sf2}
\end{equation}
\begin{equation}
	H_{ab} w^a w^b > 0,\quad {\rm for\ some\ } w^a,
\label{sf4}
\end{equation}
where we define
\begin{equation}
	A = \frac{1}{2}\frac{\frac{1}{n} G\; _2G - \,_3G}{_2G - \frac{1}{n}G^2},
\label{A2def}
\end{equation}
and 
\begin{equation}
	H_{ab} = G_{ab} + \frac{1}{2-n} \left(G + 2A \right) g_{ab}.
\label{H2def}
\end{equation}

\begin{proof}
We begin by showing the conditions are necessary. Suppose $(g, \psi)$ define a non-null solution to the Einstein-scalar field equations with $V(\psi)=0$.  From  the Einstein equations (\ref{esf}) we have
\begin{equation}
	 G_{ab} - \frac{1}{n} G g_{ab} =q\left( \psi_{;a} \psi_{;b} - \frac{1}{n} g^{lm} \psi_{;l} \psi_{;m} g_{ab}\right),
	 \label{Seq}
\end{equation}
and
\begin{equation}
	_2G_{ab} - \frac{1}{n}\, {}_2G g_{ab} = -{2}{q} \Lambda \left( \psi_{;a} \psi_{;b} - \frac{1}{n} g^{lm} \psi_{;l} \psi_{;m} g_{ab} \right).
\end{equation}
These  two equations imply
\begin{equation}
	\,_2G_{ab} - \frac{1}{n}\, _2G g_{ab}=-2\Lambda (G_{ab} - \frac{1}{n} G g_{ab} ).
\label{A1}
\end{equation}
Now multiplying (\ref{A1}) by $G^{ab}$ we get
\begin{equation}
	{}_3G - \frac{1}{n} G\; _2G  = -2 \Lambda ({}_2G - \frac{1}{n}G^2),
\label{A2}
\end{equation}
while contracting (\ref{Seq}) with $G^{ab}$ gives
\begin{equation}
 {}_2G - \frac{1}{n}G^2 = \frac{n-1}{n} q^2\left(g^{ab}\psi_{;a}\psi_{;b}\right)^2 \neq 0.
 \end{equation}
 Equation (\ref{A2}) then implies
 \begin{equation}
 A\equiv \frac{1}{2}\frac{\frac{1}{n} G\; _2G - {}_3G}{{}_2G - \frac{1}{n}G^2} = \Lambda.
 \end{equation}
 It follows that $A_{;i} = 0$ and then, from the Einstein equations, we have that
 \begin{equation}
 H_{ab} = q \psi_{;a}\psi_{;b}, 
 \end{equation}
 from which follow the rest of the conditions, (\ref{sf1}), (\ref{sf2}), (\ref{sf4}). 
 
 Conversely, suppose the conditions (\ref{sf5})--(\ref{sf4}) are satisfied. From Proposition \ref{curlgrad},  equations (\ref{sf1}), (\ref{sf2}), (\ref{sf4}) imply that there exists a function $\psi$ such that $H_{ab} = q\psi_{;a}\psi_{;b}$, while (\ref{sf0}) implies that $A= const.\equiv  \Lambda$.  Together, these results imply that:
 \begin{equation}
 G_{ab}  = q\left( \psi_{;a}\psi_{;b} - \frac{1}{2} g_{ab}g^{mn}\psi_{;m}\psi_{;n}\right) - \Lambda g_{ab}\,
 \end{equation}
 so that the Einstein equations are satisfied by $(g, \psi)$.
  The contracted Bianchi identity now implies
  \begin{equation}
  \psi_{;c} g^{ab}\psi_{;ab} = 0,
  \end{equation}
 and the non-null condition (\ref{sf5}) yields $g^{ab}\psi_{;a}\psi_{;b} \neq 0$, which enforces the scalar field equation (\ref{kg}).
\end{proof}
\label{nonullfreesf}
\end{Theorem}

From the proof just given it is clear the scalar field is determined from the metric by solving a system of quadratic equations followed by a simple integration.
\begin{Corollary}
Let $(M, g)$ satisfy the conditions of Theorem \ref{nonullfreesf}.  Then $(g, \psi)$ satisfy the Einstein-scalar field equations, with $V = 0$, with $\Lambda = A$, and with $\psi$ determined  up to an additive constant and up to a sign by
\begin{equation}
\psi_{;a}\psi_{;b} = \frac{1}{q}H_{ab}.
\end{equation}
\label{nonnullsfcor}
\end{Corollary}

We  turn to the special case of free, massless, null solutions, that is, solutions in which $g^{ab}\psi_{;a}\psi_{;b}=0$.
\begin{Theorem}
Let $(M, g)$ be an $n$-dimensional spacetime, $n > 2$. The following are  necessary and sufficient conditions on $g$ such that there exists a scalar field $\psi$ with $(g, \psi)$ defining a local, null solution of the Einstein-scalar field equations with $V(\psi) = 0$:
\begin{equation}
R = \frac{2n}{n-2}\Lambda,
\label{nsf0}
\end{equation}
\begin{equation}
	S_{a[b}S_{c]d} = 0,
\label{nsf1}
\end{equation}
\begin{equation}
	 S_{ab}S_{c[d;e]} + S_{ac}S_{b[d;e]} + S_{bc;[d} S_{e]a} = 0, 
\label{nsf2}
\end{equation}
\begin{equation}
	S_{ab} w^a w^b > 0 \quad{\rm for\ some\ } w^a,
\label{nsf4}
\end{equation}
where $S_{ab}$ is the trace-free Ricci tensor.
\begin{proof}
Suppose the Einstein equations (\ref{esf}) with $V(\psi)=0$ are satisfied for some $(g, \psi)$  where $g^{lm}\psi_{;l} \psi_{;m} = 0$. Taking the trace of (\ref{esf}) leads to (\ref{nsf0}) and the trace-free part yields
\begin{equation}
	S_{ab} = q\psi_{;a} \psi_{;b},
\end{equation} 
from which (\ref{nsf1}), (\ref{nsf2}), and (\ref{nsf4}) follow.   

Conversely, using Proposition \ref{curlgrad},  conditions (\ref{nsf1}), (\ref{nsf2}), and (\ref{nsf4}) imply that locally there exists a function $\psi$ such that
\begin{equation}
S_{ab} = q\psi_{;a} \psi_{;b}.
\end{equation}
Note that this implies $g^{lm}\psi_{;l} \psi_{;m} = 0$.  Using (\ref{nsf0}) and $S_{ab}$ to construct the Einstein tensor leads to the free, massless, null field Einstein equations:
\begin{equation}
G_{ab} =q \psi_{;a} \psi_{;b} - \Lambda g_{ab}.
\end{equation}
The contracted Bianchi identity, along with $g^{lm}\psi_{;l} \psi_{;m} = 0$, then implies the field equation
\begin{equation}
g^{ab} \psi_{;ab} = 0.
\end{equation}
\end{proof}
\label{nullfreesf}
\end{Theorem}

\begin{Corollary}
Let $(M, g)$ satisfy the conditions of Theorem \ref{nullfreesf}.  Then $(g, \psi)$ satisfy the Einstein-scalar field equations with $V = 0$,  and $\psi$ is determined up to an additive constant and up to a sign by
\begin{equation}
\psi_{;a}\psi_{;b} = \frac{1}{q}S_{ab}.
\end{equation}
\end{Corollary}

We now turn to geometrization conditions which describe the case with $V(\psi) \neq0$.  We define $\V = {q}V(\psi) + \Lambda$.  We always assume that $V(\psi)$ has been specified such that there exists an inverse function $W\colon {\bf R}\to {\bf R}$ with $W(\V(\psi)) = \psi$, and $\V(W(x)) = x$. 
  
\begin{Theorem}
Let $(M, g)$ be an $n$-dimensional spacetime, $n>2$. The following are  necessary and sufficient conditions on $g$ such that there exists a scalar field $\psi$  with $(g, \psi)$ defining a non-null solution to the Einstein-scalar equations (\ref{esf}), (\ref{kg}), where $\V = {q} V + \Lambda$ has inverse $W$:
\begin{equation}
 {}_2G - \frac{1}{n}G^2 \neq 0,
 \label{sfa5}
 \end{equation}
 \begin{equation}
H_{ab} = qW^{\prime2}(A)A_{;a} A_{;b},
 \label{sfa0}
 \end{equation}
where we define
\begin{equation}
	A = \frac{1}{2}\frac{\frac{1}{n} G\; _2G - \,_3G}{_2G - \frac{1}{n}G^2}.
\end{equation}
and 
\begin{equation}
	H_{ab} = G_{ab} + \frac{1}{2-n} \left(G + 2A \right) g_{ab}.
\label{HVdef}
\end{equation}
\label{sfV}
\end{Theorem}

\begin{proof}
The proof is along the same lines as the proof of Theorem \ref{nonullfreesf}.   To see that the conditions  (\ref{sfa5}), (\ref{sfa0}) are necessary, we start from the Einstein equations (\ref{esf}), from which it follows  that the scalar field is non-null only if
\begin{equation}
{}_2G - \frac{1}{n}G^2 \neq 0.
\end{equation}
It also follows from (\ref{esf}) that
 \begin{equation}
 A\equiv \frac{1}{2}\frac{\frac{1}{n} G\; _2G - {}_3G}{{}_2G - \frac{1}{n}G^2} = \V(\psi),
 \label{AVdef}
 \end{equation}
 so that
 \begin{equation}
\psi = W(A),
 \end{equation} 
 and
 \begin{equation}
 H_{ab} = q\psi_{;a} \psi_{;b} = q W^{\prime2}(A)A_{;a} A_{;b}.
 \end{equation}
 Conversely, assuming the metric satisfies  conditions (\ref{sfa5}) and (\ref{sfa0}), if we set 
\begin{equation}
\psi \equiv W(A),
\end{equation}
so that $A = \V(\psi)$, then from (\ref{sfa5}), (\ref{sfa0})  the scalar field is not null and the Einstein equations are satisfied.   The contracted Bianchi identity then implies the scalar field equation (\ref{kg}) is satisfied as before. 
\end{proof}

\begin{Corollary}
If a metric $g$ satisfies the conditions of Theorem \ref{sfV}, then there exists a non-null solution $(g, \psi)$ to the Einstein-scalar field equations (\ref{esf}), (\ref{kg}) where
\begin{equation}
\psi = \V^{-1}(A).
\end{equation}
\label{CV}
\end{Corollary}

Finally we consider the null case with a given self-interaction potential, invertible as before.
\begin{Theorem}
Let $(M, g)$ be an $n$-dimensional spacetime, $n>2$. There exists a scalar field $\psi$  with $(g, \psi)$ defining a local, null solution to the Einstein-scalar equations (\ref{esf}), (\ref{kg}) (with potential described by $\V = qV + \Lambda$ and  $W = \V^{-1}$) if and only if either 
\begin{equation}
S_{ab} = \q W^{\prime 2}(B) B_{;a} B_{;b} \neq 0,
\label{sfnullV}
\end{equation}
or
\begin{equation}
S_{ab} = 0,\quad B_{;a}=0,\quad V^\prime(W(B)) = 0,
\label{sfnullV2}
\end{equation}
where
\begin{equation}
B =\frac{n-2}{2n}R.
\end{equation}

\label{nsfthm}
\end{Theorem}

\begin{proof}
To see that this condition is necessary, assume the Einstein-scalar field equations hold for a metric $g$ and a null scalar field $\psi$.  It follows that
\begin{equation}
\V = \frac{n-2}{2n} R,\quad S_{ab} = q \psi_{;a}\psi_{;b},
\end{equation}
so that $\psi = W(B)$ and condition (\ref{sfnullV}) or condition (\ref{sfnullV2}) follows, depending upon whether $\psi_{;a}$ vanishes or not.  Conversely, defining $\psi = W(B)$, it follows that $R = \frac{2n}{n-2} \V(\psi)$ and, using (\ref{sfnullV}) if $S_{ab}\neq 0$, the Einstein equations (\ref{esf}) are satisfied and the scalar field is null.   The contracted Bianchi identity implies
\begin{equation}
\psi_{;b}\left[\psi_{;a}{}^a - V^\prime(\psi)\right] = 0,
\end{equation}
which implies (\ref{kg}) if $S_{ab}\neq0$ since $\psi_{;b}\neq 0$ by (\ref{sfnullV}). If $S_{ab}=0$, and $B_{;a} = 0$, the scalar field equation follows from $V^\prime(W(B)) = 0$.
\end{proof}

\begin{Corollary}
If a metric $g$ satisfies the conditions of Theorem \ref{nsfthm}, then there exists a null solution $(g, \psi)$ to the Einstein-scalar field equations (\ref{esf}), (\ref{kg}) where
\begin{equation}
\psi = \tilde V^{-1}(B).
\end{equation}
\end{Corollary}

\subsection{Example: A non-inheriting scalar field solution}

The static, spherically symmetric fluid spacetime (\ref{fluid1}), (\ref{fluid1b}) also satisfies the geometrization conditions for a massless free scalar field contained in Theorem \ref{nonullfreesf}.  Starting with the metric
\begin{equation}
g = -r^2 dt \otimes dt +\frac{2}{1 + \lambda r^2} dr \otimes dr + r^2(d\theta \otimes d\theta + \sin^2\theta d\phi \otimes d\phi).
\label{fluid1aa}
\end{equation}
and calculating $A$ and $H$ from (\ref{A2def})  and  (\ref{H2def})  gives
\begin{equation}
 A = -\frac{3}{2} \lambda, \quad H_{ab} dx^a \otimes dx^b = dt \otimes dt,
\end{equation}
so that, according to Corollary \ref{nonnullsfcor}, the scalar field is given by 
\begin{equation}
\psi = \pm\, \frac{1}{\sqrt{q}}t + constant,
\end{equation} 
and the cosmological constant is given by $\Lambda = - \frac{3}{2} \lambda$.   This solution (with $\lambda = 0$) was exhibited in ref.~\cite{Wyman}.  We remark that while the spacetime is static the scalar field is clearly not static and so represents an  example of a ``non-inheriting'' solution to the Einstein-scalar field equations.  Non-inheriting solutions of the Einstein-Maxwell equations are well-known \cite{Stephani}.  Geometrization conditions, which depend solely upon the metric, treat inheriting and non-inheriting matter fields on the same footing.  

We have been able to find an analogous family of non-inheriting solutions in 2+1 dimensions from an analysis of the geometrization conditions in Theorem \ref{nonullfreesf}. In coordinates $(t, r, \theta)$ the spacetime metric takes the form:
\begin{equation}
g = - \frac{1}{2\Lambda} dt \otimes dt + \frac{1}{b - 2\Lambda r^2} dr \otimes dr + r^2 d\theta,
\end{equation}
where $b$ is a constant.  This metric yields $A = \Lambda$ and
\begin{equation}
H =  dt \otimes dt,
\end{equation}
so that from Corollary \ref{nonnullsfcor} the scalar field is given by
\begin{equation}
\psi = \pm \frac{1}{\sqrt{q}} t + constant.
\end{equation}
It is straightforward to verify that the metric and scalar field so-defined satisfy the Einstein-scalar field equations (\ref{esf}) and  (\ref{kg}) with $V=0$. 

\subsection{Example: No-go results for spherically symmetric null scalar field solutions}

We use Theorem \ref{nullfreesf} to show that there are no null solutions to the free, massless Einstein-scalar field equations if the spacetime is static and spherically symmetric, provided the spherical symmetry orbits are not null.  We also show that there are no spherically symmetric null solutions with null spherical symmetry orbits. Both results hold with or without a cosmological constant.  Since these results follow directly from  the geometrization conditions they apply whether or not the scalar field inherits the spacetime symmetries.

We first consider a static, spherically symmetric spacetime in which the spherical symmetry orbits are not null. We use coordinates chosen such that the metric takes the form:
\begin{equation}
g = - f(r) dt \otimes dt + h(r) dr \otimes dr + R^2(r) (d\theta \otimes d\theta + \sin^2\theta d\phi \otimes d\phi),
\label{nonnullmetric}
\end{equation}
for some non-zero functions $f, h, R$.  The condition (\ref{nsf1}) applied to (\ref{nonnullmetric})  yields:
\begin{equation}
R^{\prime\prime} = \frac{1}{2} \left(\frac{1}{h} h^\prime +  \frac{1}{f} f^\prime\right)R^\prime  ,
\end{equation}
\begin{equation}
f^{\prime\prime} = 2 f\left( \frac{R^\prime}{R}\right)^2 + \frac{1}{2h} h^\prime f^\prime + \frac{1}{2f} f^{\prime2} - 2\frac{fh}{R^2}.
\end{equation}
These conditions force the trace-free Ricci tensor to vanish, whence the scalar field vanishes and we have an Einstein space.  Consequently there are no non-trivial null solutions to the Einstein-scalar field equations in which the spacetime is static and spherically symmetric with non-null spherical symmetry orbits.

Next we consider a spherically symmetric spacetime in which the spherical symmetry orbits are null. In this case there exist coordinates $(v, r, \theta, \phi)$ such that the metric takes the form:
\begin{equation}
g = w(v, r) (dv \otimes dr + dr \otimes dv) + u(v,r) dr \otimes dr +  r^2 (d\theta \otimes d\theta + \sin^2\theta d\phi \otimes d\phi),
\label{nullmetric}
\end{equation}
for some functions $w\neq0$ and $u$.  
Calculation of  conditions (\ref{nsf0}), (\ref{nsf1}) for metrics (\ref{nullmetric}) reveals they are are incompatible.  Consequently there are no null solutions to the Einstein-scalar field equations in this case.  Since (\ref{nullmetric}) is not actually static, but merely spherically symmetric, this proves that there are no Einstein-free-scalar field null solutions for spacetimes which are spherically symmetric with null symmetry orbits. 

\subsection{Self-interacting scalar fields}

Fonarev \cite{Fonarev} has found a 1-parameter family of non-null spherically symmetric solutions to the Einstein-scalar field equations with a potential energy function which is an exponential function of the scalar field.  Here we verify these solutions directly from the metric using Theorem \ref{sfV}.

In coordinates $(t, r, \theta, \phi)$ and with $q=1$ the metric in reference \cite{Fonarev} takes the form:
\begin{align}
g = -&e^{8\alpha^2 \beta t}(1-\frac{2m}{r})^\delta dt \otimes dt
+e^{2\beta t}(1-\frac{2m}{r})^{-\delta} dr \otimes dr\nonumber \\
+& e^{2\beta t}(1 - \frac{2m}{r})^{1-\delta} r^2(d\theta \otimes d\theta + \sin^2\theta d\phi \otimes d\phi),
\label{Fonarev}
\end{align}
where $m > 0$ is a free parameter, 
\begin{equation}
\delta = \frac{2\alpha}{\sqrt{4\alpha^2 + 1}},`
\end{equation}
and $\alpha$ and $\beta$ parametrize the scalar field potential and cosmological constant via
\begin{equation}
\tilde V(\psi) =  \beta^2(3-4\alpha^2)\exp\left(-\sqrt{8}\alpha\psi\right).
\label{FV}
\end{equation}
The inverse of the potential function is given by
\begin{equation}
W(x) = - \frac{\sqrt{2}}{4\alpha} \ln\left(\frac{x}{\beta^2(3 - 4\alpha^2)}\right).
\end{equation}

Using the metric (\ref{Fonarev}) to calculate $A$ in (\ref{AVdef}) gives
\begin{equation}
A = (3-4\alpha^2) \beta^2e^{-8\alpha^2 \beta t}\left(1 - \frac{2m}{r}\right)^{\frac{2\alpha}{\sqrt{4\alpha^2 + 1}}}.
\end{equation}
Calculating the tensor $H$ in (\ref{HVdef}) yields
\begin{align}
H =& 8\beta^2\alpha^2 dt \otimes dt + \frac{4\alpha m \beta}{\sqrt{4\alpha^2 + 1}r(r - 2m)} (dt \otimes dr + dr \otimes dt)\nonumber\\ 
&+ \frac{2m^2}{(4\alpha^2 + 1)r^2(r - 2m)^2} dr \otimes dr,
\end{align}
and it follows that (\ref{sfa0}) is satisfied. Therefore, from Theorem \ref{sfV}, the metric (\ref{Fonarev}) does indeed define a scalar field solution with the potential (\ref{FV}).  Using Corollary \ref{CV}, the scalar field is calculated to be
\begin{equation}
\psi = \sqrt{2}\left\{2 \alpha \beta t + \frac{1}{2\sqrt{4\alpha^2 + 1}}\ln\left(1- \frac{2m}{r}\right)\right\},
\end{equation}
in agreement with Fonarev \cite{Fonarev}.

\section{Electromagnetic Fields}

Necessary and sufficient conditions for a metric  to be a non-null electrovacuum were first given by Rainich \cite{Rainich} and enhanced by Misner and Wheeler \cite{MW}.   Necessary and sufficient conditions for a metric to be a null electrovacuum have been obtained in \cite{Torre2014}.  In each case procedures for constructing the electromagnetic field from the metric have also been obtained. These results apply to the Einstein-Maxwell equations in four spacetime dimensions with no electromagnetic sources and with no cosmological constant.   Here we summarize all these results while generalizing them to include a cosmological constant. 

The Einstein-Maxwell equations for a spacetime $(M, g)$ with electromagnetic field $F_{ab} = F_{[ab]}$ and cosmological constant $\Lambda$ are given by 

\begin{equation}
R_{ab} - \half R g_{ab}  + \Lambda g_{ab} = \q \left(F_{ac}F_b{}^c - \frac{1}{4} F_{cd}F^{cd} g_{ab}\right),
\label{EE}
\end{equation}

\begin{equation}
\nabla^b F_{ab} = 0,\quad  \nabla_{[a}F_{bc]} = 0.
\label{ME}
\end{equation}
Here $\q=8\pi \G/c^4$, $\nabla$ is the torsion-free derivative determined by the metric, $R_{ab}$ and $R$ are the Ricci tensor and Ricci scalar of the metric.  We shall refer to (\ref{EE}) alone as the Einstein equations and we shall refer to (\ref{ME}) alone as the Maxwell equations.    

If the 2 scalar invariants of the electromagnetic field vanish in some region,
\begin{equation}
F_{ab} F^{ab} = 0 = \epsilon^{abcd} F_{ab} F_{cd},
\end{equation}
we say that the electromagnetic field is {\it null} in that region.  Otherwise the electromagnetic field is {\it non-null} in that region.  If $\zeta_{ab}$ is a 2-form, the {\it Hodge duality} operation is given by
\begin{equation}
{}^\star\zeta_{ab} = \frac{1}{2} \epsilon_{ab}{}^{cd}\zeta_{cd}.
\end{equation}

\begin{Proposition}{
Let $(M, g)$ be a 4-dimensional spacetime.  The following conditions on $g$ are necessary and sufficient for the existence of a 2-form $F$ such that the Einstein equations (\ref{EE}) are satisfied on $M$:
\begin{equation}
R =4\Lambda, \quad S_a^b S_b^c - \frac{1}{4} \delta_a^c S_{mn}S^{mn} = 0, \quad S_{ab}t^at^b >0,
\label{AlgRain}
\end{equation}
where $S_{ab} = R_{ab} - \frac{1}{4} R g_{ab}$ and $t^a$ is any timelike vector field. 
The 2-form $F_{ab}$ which exists when the conditions (\ref{AlgRain}) are satisfied by the metric is unique up to a local duality rotation:
\begin{equation}
F_{ab} \longrightarrow \cos(\phi) F_{ab} - \sin(\phi){}^\star F_{ab},
\label{Fgen}
\end{equation}
where $\phi\colon M \to {\bf R}$. 

}
\label{P1}
\end{Proposition}

\begin{proof}

Decompose the Einstein equations (\ref{EE})  into pure trace and trace-free parts:
\begin{equation} 
R = 4\Lambda, \quad S_{ab} = \q \left(F_{ac}F_b{}^c - \frac{1}{4} F_{cd}F^{cd} g_{ab}\right).
\label{redEQ}
\end{equation}
It is straightforward to check that these equations imply conditions (\ref{AlgRain}), so these conditions are necessary.  The algebraic results of \cite{MW}, applied to $S_{ab}$ (instead of $R_{ab})$ show that the conditions on $S_{ab}$ in (\ref{AlgRain}) are sufficient for the existence of 2-form $F$ satisfying the Einstein equations, with  $F$  determined by the metric up to a duality rotation (\ref{Fgen}).
\end{proof}


The conditions (\ref{AlgRain}) generalize the classical algebraic Rainich conditions; they apply equally well for non-null and null electromagnetic fields. In the null case they take a simpler form.

\begin{Corollary}{
Let $(M, g)$ be a 4-dimensional spacetime.  The following conditions on $g$ are necessary and sufficient for the existence of a null  2-form  $F$ such that the Einstein equations (\ref{EE}) are satisfied on $M$:
\begin{equation}
R =4 \Lambda, \quad S_a^b S_b^c  = 0, \quad S_{ab}t^at^b >0,
\label{AlgRain2}
\end{equation}
}
\label{C1}
\end{Corollary}

\begin{proof}

With the substitution of $S_{ab}$ for $R_{ab}$, the proof is identical to that found in \cite{MW}.

\end{proof}

\begin{Corollary}{
Let $(M, g)$ be a 4-dimensional spacetime.  There exists a null electromagnetic 2-form  $F$ such that the Einstein equations (\ref{EE}) are satisfied if and only if $R = 4\Lambda$ and there exists a null geodesic congruence with tangent vector field $k^a$ such that
\begin{equation}
 S_{ab} = \frac{1}{4} k_a k_b,\quad g_{ab} k^a k^b = 0.
\label{K}
\end{equation}
}
\label{C2}
\end{Corollary}

\begin{proof}
This follows from Corollary \ref{C1}. Since $S_a^b S_b^c  = 0$, eq. (\ref{K}) is proved exactly as in \cite{MW}.  That $k^a$ is tangent to geodesics follows from $R=const.$ and the contracted Bianchi identity. 
\end{proof}

From Proposition \ref{P1} and the general form of the electromagnetic 2-form in (\ref{Fgen}), the classical results of Rainich, Misner and Wheeler generalize to the Einstein-Maxwell equations with a cosmological constant as follows.

\begin{Theorem}{
Let $(M, g)$ be a 4-dimensional spacetime.  There exists a non-null 2-form $F_{ab}$ such that $(g, F)$ satisfy the Einstein-Maxwell equations (\ref{EE}), (\ref{ME}) if and only if $g$ satisfies:
 
 \begin{equation}
R =4\Lambda, \quad S_a^b S_b^c = \frac{1}{4} \delta_a^c S_{mn}S^{mn} \neq 0, \quad S_{ab}t^at^b >0,
\label{AlgRainNonNull}
\end{equation}

\begin{equation}
\nabla_{[a}\left(\epsilon_{b]cde}{S^c_m\nabla^dS^{me}\over S_{ij}S^{ij}}\right) = 0.
\label{DiffRain}
\end{equation}
}
\label{Rainich}
\end{Theorem}

\begin{proof}
After decomposing the Einstein equations into pure trace and trace-free parts, (see (\ref{redEQ})), the proof is the same is in \cite{MW} with $R_{ab}$ replaced by $S_{ab}$.
\end{proof}
     From the proofs of Proposition \ref{Fgen} and Theorem \ref{Rainich}  we have the following formulas for constructing the electromagnetic field from the metric \cite{MW}.
 
\begin{Corollary}
{
Let $(M, g)$ be a 4-dimensional spacetime satisfying the conditions of Theorem \ref{Rainich}.  Then the fields $(g, F)$ satisfy the Einstein-Maxwell equations with
\begin{equation}
F_{ab}= \cos(\phi)\xi_{ab} - \sin(\phi){}^\star\xi_{ab},
\end{equation}
where $\xi_{ab}$ is a solution to
\begin{equation}
\xi_{ab}\xi_{cd}  = \frac{1}{\q} \left(E_{abcd} - [S_{mn}S^{mn}]^{-1/2} E_{abef}E_{cd}{}^{ef}\right),\quad \xi_{ab}{}^\star \xi^{ab}=0,\quad \xi_{ab}\xi^{ab} < 0,
\end{equation}
\begin{equation}
E_{abcd} = \frac{1}{2}\left(g_{ac}S_{bd} - g_{bc}S_{ad} + g_{bd}S_{ac} - g_{ad}S_{bc}\right),
\end{equation}and $\phi$ is a solution to
\begin{equation}
\nabla_b\phi = {\epsilon_{bcde} S^c_m\nabla^dS^{me}\over S_{ij}S^{ij}}.
\end{equation}
}
\label{nonnullcor}
\end{Corollary}
To give analogous necessary and sufficient conditions for null electrovacua we define  a {\it null tetrad adapted to the vector field $k$} to be a null tetrad $(k, l, m, \overline m)$, where $k$ and $l$ are real and  where $m$ and $\overline m$ are complex conjugates. The null tetrad satisfies
\begin{equation}
g_{ab} k^a l^b = -1,\quad g_{ab} m^a\overline m^b = 1,
\end{equation}
with all other scalar products vanishing. We will use the Newman-Penrose formalism for this tetrad as defined, {\it e.g.,} in \cite{Stewart}. In particular, we recall the definitions of the following Newman-Penrose spin coefficients.  The twist of the null congruence with tangent field $k^a$ is given by
\begin{equation}
\omega =  \Im\left(\overline m^a m^b \nabla_a k_b\right).
\end{equation}
The shear of the null congruence is defined by 
\begin{equation}
\sigma =  \left( m^a m^b \nabla_a k_b\right).
\end{equation}
The acceleration of the congruence is defined by
\begin{equation}
\kappa =  \left( k^b m^a \nabla_b k_a\right).
\end{equation}
The quantities $\omega$, $|\sigma|$, and $|\kappa|$ are intrinsic properties of the null congruence determined by $k^a$ and are independent of the choice of adapted tetrad.  The congruence is surface forming if and only if $\omega = 0$. The congruence consists of geodesics precisely when $\kappa = 0$.  

Finally, we denote by $\db^\alpha = (\db_k, \db_l, \db_m, \db_{\overline m})$ the basis of 1-forms dual to $e_\alpha = (k, l, m, \overline m)$, $\alpha = 1,2,3,4$.  The dual basis satisfies
\begin{equation}
\db^\alpha(e_\beta) = \delta^\alpha_\beta .
\end{equation}

\begin{Theorem}{
Let $(M, g)$ be a 4-dimensional spacetime.  The following three conditions on $g$ are necessary and sufficient  for the existence of a null 2-form $F_{ab}$ such that $(g, F)$ satisfy the Einstein-Maxwell equations (\ref{EE}), (\ref{ME}):

\medskip
\item{\rm (1)} The scalar curvature is constant and the trace-free Ricci tensor is null and positive,
\begin{equation}
R =4 \Lambda, \quad S_a^b S_b^c  = 0\  \Longleftrightarrow\  S_{ab} = \frac{1}{4}k_a k_b, \quad S_{ab}t^at^b >0.
\end{equation}
\smallskip
\item{\rm (2)} The null congruence with tangent field $k^a$ is geodesic and  shear-free.\footnote{Conditions (\ref{sfng}) correspond to the Mariot-Robinson theorem \cite{MR}, which states that the repeated principal null direction of a null electromagnetic field is necessarily tangent to a shear-free geodesic null congruence, and is also a repeated principal null direction of the algebraically special spacetime upon which the null electromagnetic field resides.  This condition is necessary and sufficient (at least in the analytic setting) for the existence of a null solution to the Maxwell equations on a given spacetime, but is only a necessary condition for a  solution to the Einstein-Maxwell equations with null electromagnetic field.}
\begin{equation}
\kappa = 0 = \sigma.
\label{sfng}
\end{equation}
\smallskip
\item{\rm (3)}
If the twist of the null congruence defined by $k$ vanishes, $\omega = 0$, then 
\begin{equation} 
\Re\left[(\overline\delta+\overline\beta-\alpha)(\tau-2\beta)\right] - \half (\mu - \overline\mu)(\epsilon-\overline\epsilon) = 0,
\label{IC4}  
\end{equation}
or, if the twist is non-vanishing, $\omega\neq 0$,
\begin{align}
&{\omega}\delta \Big\{\Re\Big[\delta(\overline\tau-2\overline\beta)\Big]\Big\} - \Big[\delta \omega + \omega(\tau - \overline\alpha - \beta)\Big]\Big[ \Re\left\{(\overline\delta + \overline\beta - \alpha)(\tau - 2\beta)\right\}\nonumber\\
&+ i\Im(\mu)(\rho - 2\epsilon)\Big]
+\frac{\omega}{2}\Big\{\overline\beta\delta(2\overline\alpha + \tau - 4\beta) + \beta\delta(2\alpha + \overline\tau - 4\overline\beta) + 2i\delta(\Im(\mu)(\rho-2\epsilon))\nonumber\\ 
&+ \tau\delta(\overline\beta - \alpha) +  \overline\tau\delta(\beta - \overline\alpha)
- \alpha\delta(\tau - 2\beta) - \overline\alpha\delta(\overline\tau - 2\overline\beta)\Big\}
- i\omega^2\Delta(\tau - 2\beta)\nonumber\\
&+\omega^2\Big[\overline\nu(\omega + 2\Im(\epsilon)) -(\tau-2\beta)(2\Im(\gamma) + i\mu)  +i\overline\lambda (\overline\tau - 2\overline\beta) \Big] = 0,
\label{IC1}
\end{align}
\medskip
Conditions (2) and (3) use the Newman-Penrose quantities for any null tetrad adapted to $k^a$.}
\label{NEV}
\end{Theorem}

\begin{proof}

The proof is a straightforward generalization of that given in \cite{Torre2014}.  Here are the salient points.
Setting
\begin{equation}
s^a = \frac{1}{\sqrt{2}}(m^a + \overline m^a),\quad \xi_{ab} = \frac{1}{\sqrt{q}} k_{[a}s_{b]},
\end{equation}
the electromagnetic 2-form $F$ solving the Einstein equations is, according to Proposition \ref{P1},  given by
\begin{equation}
F_{ab} = \cos(\phi)\xi_{ab} - \sin(\phi) {}^\star\xi_{ab},
\end{equation}
where $\phi\colon M\to {\bf R}$ is any function.  The Maxwell equations for $F_{ab}$ are equivalent to (\ref{sfng})  and a system of differential equations for $\phi$ given by
\begin{equation}
d\phi =   i(2\epsilon - \rho) \db_k + i(2\beta - \tau)\db_m - i(2\overline\beta - \overline\tau)\db_{\overline m}.
\label{dphieq}
\end{equation}
Here again we use the Newman-Penrose spin coefficients as defined, {\it e.g.,} in \cite{Stewart}.  The integrability conditions for these equations are computed precisely as in \cite{Torre2014} -- the cosmological constant does not enter into the computation. In the twisting case ($\omega \neq0$) there are six conditions.  Four of these six conditions vanish by virtue of the Ricci identities -- a result which holds precisely as shown in \cite{Torre2014}, even allowing for a non-vanishing cosmological constant in the relevant Ricci identities. This leaves  two non-trivial  conditions, which are shown as one complex condition in (\ref{IC1}). In the twist-free case ($\omega = 0$) there are 3  integrability conditions for (\ref{dphieq}).  Two of these conditions are trivially satisfied as a consequence of the Ricci identities. The cosmological constant does enter into the Ricci identities used in this case.  The remaining, non-trivial condition is (\ref{IC4}). \end{proof}

We remark that condition (\ref{IC4}) is one real condition depending upon as many as four derivatives of the metric.  Condition (\ref{IC1}) is complex; it represents two real conditions and depends upon as many as five derivatives of the metric.  In either case, the conditions for an electrovacuum described in Theorem \ref{NEV} do not depend on the choice of tetrad adapted to $k$ \cite{Torre2014}. 

From the proof of Theorem \ref{NEV}  the null electromagnetic field is constructed from the metric as follows. 

\begin{Corollary}{
Let $(M, g)$ be a 4-dimensional spacetime satisfying the conditions of Theorem \ref{NEV}.  Then the fields $(g, F)$ satisfy the Einstein-Maxwell equations with
\begin{equation}
F_{ab}= \cos(\phi)\xi_{ab} - \sin(\phi){}^\star\xi_{ab},
\end{equation}
where
\begin{equation}
\xi_{ab} = \frac{1}{\sqrt{q}}\,k_{[a}s_{b]},\quad s^a = \frac{1}{\sqrt{2}}(m^a + \overline m^a),
\end{equation}
and $\phi$ is a solution to
\begin{equation}
d\phi =   i(2\epsilon - \rho) \db_k + i(2\beta - \tau)\db_m - i(2\overline\beta - \overline\tau)\db_{\overline m}.
\end{equation}
}
\label{NEVcor}
\end{Corollary}

We remark that, just as in the case with vanishing cosmological constant, $F$ is determined up to a duality rotation in the twisting case, and $F$ involves an arbitrary function of one variable  in the twist-free case.

\subsection{Example: LBR Solution}

There is a class of solutions to the Einstein-Maxwell equations with a spacetime that is the product of two-dimensional constant curvature spaces, which is due to Levi-Civita \cite{LeviCivita}, Bertotti \cite{Bertotti}, and Robinson \cite{Robinson1959}; we call it the LBR solution.   As an illustration of Theorem \ref{Rainich} we give a novel derivation of the LBR solution in its most general form using symmetry methods.  

Consider the set of spherically symmetric spacetimes.  As is well known (see, {\it e.g.} \cite{Stephani}), the set of such spacetimes can be partitioned into three classes according to whether the ``warp factor'' has a spacelike, null, or vanishing gradient. We consider here the latter case in which the spacetime $(M, g)$ is necessarily the product of two-dimensional geometries. We denote this product by $M = L\times S$, where $S$ has constant positive curvature and Riemannian signature (these are the spherical symmetry group orbits), and $L$ has  Lorentzian signature.  Using standard spherical polar coordinates on $S$ and null coordinates on $L$,  the spacetime metric takes the form:
\begin{equation}
g = -e^{h(u,v)} (du \otimes dv + dv \otimes du) + r_0^2(d\theta \otimes d\theta + \sin^2\theta\, d\phi \otimes d\phi).
\label{LCBR}
\end{equation}
Here  $r_0 >0$ is a constant.
We now consider the conditions of Theorem \ref{Rainich}.  We begin by imposing the condition that the scalar curvature is constant, $R =4\Lambda$.  This condition takes the form
\begin{equation}
2 e^{-h} \frac{\partial^2 h}{\partial u\partial v} + \frac{2}{r_0^2} = 4\Lambda.
\label{Lc}
\end{equation}
The first term on the left-hand side represents the scalar curvature of $L$, so this condition forces $L$ to have constant curvature, as might have been expected from the product form of the spacetime geometry.  Thus there are 3 types of  solutions to this geometrization condition according to whether the constant curvature of $L$  is positive, negative, or zero.  
Spacetimes satisfying (\ref{Lc}) are symmetric spaces; the curvature tensor is covariantly constant.  

Remarkably, the condition (\ref{Lc}) is also sufficient for the spacetime to be an electrovacuum.  It is straightforward to verify that the trace-free Ricci tensor of the metric (\ref{LCBR}), given by
\begin{equation}
S =-\frac{1}{2}\left(e^{-h}\frac{\partial^2 h}{\partial u\partial v} - \frac{1}{r_0^2}\right)  \Big[e^h(du \otimes dv + dv \otimes du) + r_0^2(d\theta \otimes d\theta + \sin^2\theta\, d\phi \otimes d\phi)\Big],
\end{equation}
satisfies (\ref{AlgRainNonNull}) and (\ref{DiffRain}) in Theorem \ref{Rainich}.  
As an example, here is a spacetime metric satisfying (\ref{Lc}) where $L$ has constant negative curvature -$4k^2$:
\begin{align}
g_{-} &= -\frac{1}{\cosh^2(k(u+v))} (du \otimes dv + dv \otimes du) + r_0^2(d\theta \otimes d\theta + \sin^2\theta\, d\phi \otimes d\phi), \label{kneg}
\\
\Lambda &= \frac{1}{2r_0^2} - k^2.\nonumber
\end{align}
Here is the electromagnetic field $F_-$, constructed from the metric $g_-$ in (\ref{kneg}) via Corollary \ref{nonnullcor}:
\begin{equation}
\sqrt{q}\,F_- = \sqrt{1 + 2k^2r_0^2} \left(\frac{\cos\alpha}{r_0\cosh^2(k(u+v))} du \wedge dv + r_0\sin\alpha \sin\theta d\theta \wedge d\phi\right),
\end{equation}
where $\alpha\in [0, 2\pi)$.
The fields $(g_-, F_-)$ satisfy the Einstein-Maxwell equations with cosmological constant $\Lambda = \frac{1}{2r_0^2} - k^2$.

\subsection{Example: A stationary null electrovacuum with  cosmological constant}


We use Theorem \ref{NEV} to construct a new hypersurface homogeneous solution to the Einstein-Maxwell equations with a cosmological constant and with a null electromagnetic field.  We begin by defining a class of metrics, parametrized by an arbitrary function $f$, which generalizes a  homogeneous null electrovacuum with cosmological constant due to Siklos \cite{Siklos}.  In coordinates $(u, v, y, z)$ our generalization is
\begin{equation}
g = \frac{3}{|\Lambda|y^2}\left[\frac{1}{2}(du \otimes dv + dv \otimes du) + dy \otimes dy + dz\otimes dz - f(y) dv \otimes dv\right],
\end{equation}
where $\Lambda < 0$.  This metric satisfies $R = 4\Lambda$. The trace-free Ricci tensor takes the pure radiation form
\begin{equation}
S = \left(\frac{1}{2} f^{\prime\prime} - \frac{1}{y} f^\prime\right)  dv \otimes dv,
\end{equation} 
with a twist-free, shear-free null geodesic generator  
\begin{equation}
k = \frac{2}{3}\Lambda y^{3/2}\sqrt{2(yf^{\prime\prime}-2 f^\prime)}\partial_u.
\end{equation}  
From Theorem \ref{NEV}, eq. (\ref{IC4}), the metric defines a  null electrovacuum if and only if $f(y)$ satisfies
\begin{equation}
y^4[f^{\prime\prime\prime2}-f^{\prime\prime}f^{\prime\prime\prime\prime}] + 2y^3[f^\prime f^{\prime\prime\prime\prime} - f^{\prime\prime\prime}f^{\prime\prime}] -2y^2f^{\prime\prime2}+12yf^{\prime\prime}f^\prime-12f^{\prime2} = 0.
\end{equation}
The Siklos solution has $f(y) = y^4$.  A new, non-trivial solution to this equation which leads to a metric satisfying all of the hypothesis of Theorem \ref{NEV} in the twist-free case is given by 
\begin{equation}
f(y) = - \frac{1}{3}y^3 +  e^y( y^2 -2 y+2).
\end{equation}
From Corollary \ref{NEVcor} the null electromagnetic 2-form $F$ which serves as the source for this spacetime is given by
\begin{equation}
\sqrt{q} F = \sqrt{\frac{3}{2|\Lambda|}}e^{y/2}\left[\cos(\frac{z}{2} - \alpha(v)) dv\wedge dy - 
\sin(\frac{z}{2} - \alpha(v)) dv\wedge dz\right],
\end{equation}
where $\alpha(v)$ is an arbitrary function. 

This null electrovacuum has Petrov type N and admits four Killing vector fields, $(\partial_u, \partial_v, \partial_z, 2z\partial_u - v\partial_z)$, generating an isometry group with three-dimensional timelike orbits and null rotation isotropy.\footnote{The Siklos solution, $f(y) = y^4$, yields a spacetime which has five Killing vector fields generating a transitive isometry group.}   The electromagnetic field does not inherit all the spacetime symmetry; for example, $F$ is not translationally invariant in the $z$ coordinate. 

\section*{Acknowledgement}
{The authors gratefully acknowledge fruitful discussions with Ian Anderson and L\"ag Avulin. This work was supported in part by grant OCI-1148331 from the National Science Foundation.}

\end{document}